\documentclass[12pt, draftclsnofoot, onecolumn]{IEEEtran}
\usepackage{amsmath,amsfonts}
\usepackage{algorithmic}
\usepackage{algorithm}
\usepackage{array}
\usepackage[caption=false,font=normalsize,labelfont=sf,textfont=sf]{subfig}
\usepackage{textcomp}
\usepackage{stfloats}
\usepackage{url}
\usepackage{verbatim}
\usepackage{graphicx}
\usepackage{cite}
\usepackage{epstopdf}
\usepackage{enumerate}
\usepackage{bm}
\usepackage{textcomp}
\usepackage{amsthm}
\usepackage{float}
\usepackage{booktabs}
\usepackage{amssymb}
\newtheorem{lemma}{Lemma}
\newtheorem{Theorem}{Theorem}

\usepackage{tikz}
\usetikzlibrary{positioning} 

\usepackage{tikz,xcolor}
\usepackage[implicit=false]{hyperref}

\hypersetup{hidelinks,
	colorlinks=true,
	allcolors=black,
	pdfstartview=Fit,
	breaklinks=true}
\definecolor{lime}{HTML}{A6CE39}
\DeclareRobustCommand{\orcidicon}{
	\begin{tikzpicture}
		\draw[lime, fill=lime] (0,0)
		circle[radius=0.16]
		node[white]{{\fontfamily{qag}\selectfont \tiny \.{I}D}};
	\end{tikzpicture}
	\hspace{-2mm}
}
\foreach \x in {A, ..., Z}{%
	\expandafter\xdef\csname orcid\x\endcsname{\noexpand\href{https://orcid.org/\csname orcidauthor\x\endcsname}{\noexpand\orcidicon}}
}

\begin{document}

\title{Why Shape Coding? Asymptotic Analysis of the Entropy Rate for Digital Images}

\author{Gangtao Xin\hspace{-2.25mm}\orcidA{}\hspace{-2.25mm}, Pingyi Fan\hspace{-2.25mm}\orcidB{}\hspace{-2.25mm},~\IEEEmembership{Senior Member,~IEEE} \\
	and Khaled B. Letaief\hspace{-2.25mm}\orcidC{}\hspace{-2.25mm}, ~\IEEEmembership{Fellow,~IEEE}

\thanks{This work was supported by the National Key Research and Development Program of China (Grant NO.2021YFA1000500(4)) }
\thanks{Gangtao Xin and Pingyi Fan are with the Department of Electronic Engineering, Tsinghua University, Beijing 100084, China, and also with the Beijing National Research Center for Information Science and Technology, Tsinghua University, Beijing 100084, China (e-mail:fpy@tsinghua.edu.cn)}
\thanks{Khaled Ben Letaief is with the School of Engineering, The Hong Kong University of Science and Technology, Hong Kong}
}

\markboth{Journal of \LaTeX\ Class Files,~Vol.~14, No.~8, August~2021}%
{Shell \MakeLowercase{\textit{et al.}}: A Sample Article Using IEEEtran.cls for IEEE Journals}

\IEEEpubid{0000--0000/00\$00.00~\copyright~2021 IEEE}

\maketitle

\begin{abstract}
This paper focuses on the limit theory of image compression. It proves that for an image source, there exists a coding method with shapes that can achieve the entropy rate under a certain condition where the shape-pixel ratio in the encoder/decoder is $\mathcal{O}({1 / {\log t}})$. Based on the new finding, an image coding framework with shapes is proposed and proved to be asymptotically optimal for stationary and ergodic processes. Moreover, the condition $\mathcal{O}({1 / {\log t}})$ of shape-pixel ratio in the encoder/decoder has been confirmed in the image database MNIST, which illustrates the soft compression with shape coding is a near-optimal scheme for lossless compression of images.

\end{abstract}


\section{Introduction and overview}

\IEEEPARstart{O}{ne} of Shannon's outstanding achievements in source coding is pointing out the data compression limit. This result has been widely and successfully applied in stream data compression. But image compression is still a challenging issue. This paper is an attempt to analyze the limit theory of image compression.

\subsection{Preliminaries}
Data compression is one of the bases of digital communications, being used to provide efficient and low-cost communication services. Image is the most important and popular medium in the current information age. Hence, image compression is naturally an indispensable part of data compression \cite{9376651}. Moreover, its coding efficiency directly affects the objective quality of the communication network and the subjective experiences of users. 

As a compression method with strict requirements, lossless image coding focuses on reducing the required number of bits to represent an image without losing quality. It guarantees to cut down the occupation of communication and storage resources as much as possible under certain system or scenario constraints. In the area of big data, image lossless coding may play a more significant role in applications where errors can not be allowed, such as in intelligent medical treatment, digital libraries, semantic communications \cite{choi2022semantic, xin2022exk}, and metaverse in the future.

Entropy rate is one of the important metrics in information theory, which extends the meaning of entropy from a random variable to a random process. It also characterizes the generalized asymptotic equipartition property of a stochastic process. In this paper, we shall employ entropy rate to explain the best achievable data compression. It is well known that the entropy rate of a stochastic process $\{Z_i\}$ is defined as
\begin{equation}
	\label{Eq:entropy rate-1}
	H(\mathcal{Z}) = \lim_{t\to \infty}  \sup {1 \over t}  H(Z_1,Z_2,...,Z_t).
\end{equation}
If the limit exists, then $H(\mathcal{Z})$ is the per symbol entropy of the $t$ random variables, reflecting how the entropy of the sequence increases with $t$. Moreover, the entropy rate can also be defined as
\begin{equation}
	\label{Eq:entropy rate-2}
	H'(\mathcal{Z})= \lim_{t \to \infty} H(Z_t|Z_{t-1},Z_{t-2},...,Z_1).
\end{equation}
$H'(\mathcal{Z})$ is the conditional entropy of the last random variable given all the past random variables. For a stationary stochastic process, the limits in Eq. (\ref{Eq:entropy rate-1}) and (\ref{Eq:entropy rate-2}) exist and are equal \cite{cover1999elements}. That is, $H(\mathcal{Z})$ = $H'(\mathcal{Z})$.
In addition, for a stationary Markov chain, the entropy rate is
\begin{align}
	\label{Eq:markov-1}
	H(\mathcal{Z})&=H'(\mathcal{Z})=\lim_{t\to \infty} H(Z_t|Z_{t-1},...,Z_1)  \\
	&=\lim_{t \to \infty} H(Z_t|Z_{t-1}).
	\label{Eq:markov-2}
\end{align}
The entropy rate is a long-term sequence metric. Even if the initial distribution of the Markov chain is not a stable distribution, it will still tend to converge as in Eq. (\ref{Eq:markov-1}) and (\ref{Eq:markov-2}).
Moreover, for a general ergodic source, the Shannon-McMillan-Breiman theorem points out its asymptotic equipartition property. If $\{Z_i\}$ is a  finite-valued stationary ergodic process, then 
\begin{equation}
	-{1 \over t} \log p(Z_0,...,Z_{t-1}) \to H(\mathcal{Z})  \text{  with probability 1}.
\end{equation}
This indicates the convergence relationship between the joint probability density and entropy rate for the general ergodic process. Following a similar idea as that of the analysis of entropy rate, we investigate the asymptotic property of shape-based coding for stationary image ergodic processes.

\subsection{Shape Coding}

A digital image is composed of lots of pixels arranged in order. This form is fixed and if the size of an image is determined, the number and arrangement mode of pixels are also determined. On the other hand, shape coding extends the basic components of images from pixels to shapes, which is a more flexible coding method and may efficiently utilize image embedding structures. Furthermore, it will no longer limit the number and position of shapes. Shape coding has three main characteristics: (1) The image is formed by filling shapes; (2) The position arrangement of shapes changes from a fixed mode to a random variable; (3) The shape database and codebook are generated in a data-driven way, which clearly contains more inherent features of image databases.

Consider a binary digital image $Z$, whose length and width are $M$ and $N$, respectively, then the total number of pixels is $t=M\times N$. Suppose it is divided into $c(t)$ shapes $\{s_1,s_2,...,s_{c(t)}\}$, where $s_i$ is the $i$-th shape. We use $\mathcal{D}$ to denote the shape database and $F_i(S_i),i=1,...,T$ to represent filling an image with shape $S_i$ at position $(x_i,y_i)$ in the $i$-th operation. The image with shape coding can be described as \cite{xin2020soft}
\begin{align}
	\min&~~\sum_{i=1}^{c(t)} [l(s_i)+l_p(x_i,y_i)] \\
	&\text{s.t.}~~Z=\sum_{i=1}^{c(t)}F_i(s_i),
\end{align}
where $l(s_i)$ and $l_p(x_i,y_i)$ represent the bit length of the shape $s_i$ and its corresponding location at $(x_i,y_i)$, respectively. The constraint condition indicates that the binary image $Z$ can be reconstructed through $c(t)$ filling operations, which is exactly the same as the original image. On this premise, shape coding tries to reduce the cost of representing an image as much as possible.

The codebook plays an important role in shape coding. It reflects the statistical characteristics and correlation of the data source. Fig. \ref{Fig:framework} illustrates the structure of shape coding. It consists of two parts: the generation and use of the codebook. On the one hand, one searches and matches the shape of images in the dataset through a data-driven method. At the same time, the frequency statistical analysis is carried out to generate a shape database. On the other hand, the codebook can be used repeatedly in communication and storage tasks to reduce the occupation of resources. The transmitter/compressor encodes the original image with the codebook. After transmission or storage through the channel/storage medium, the receiver/decompressor can decode the compressed file with the same codebook. In this way, one can completely reconstruct the original image.

\begin{figure*}
	\centerline{\includegraphics[width=6.5in]{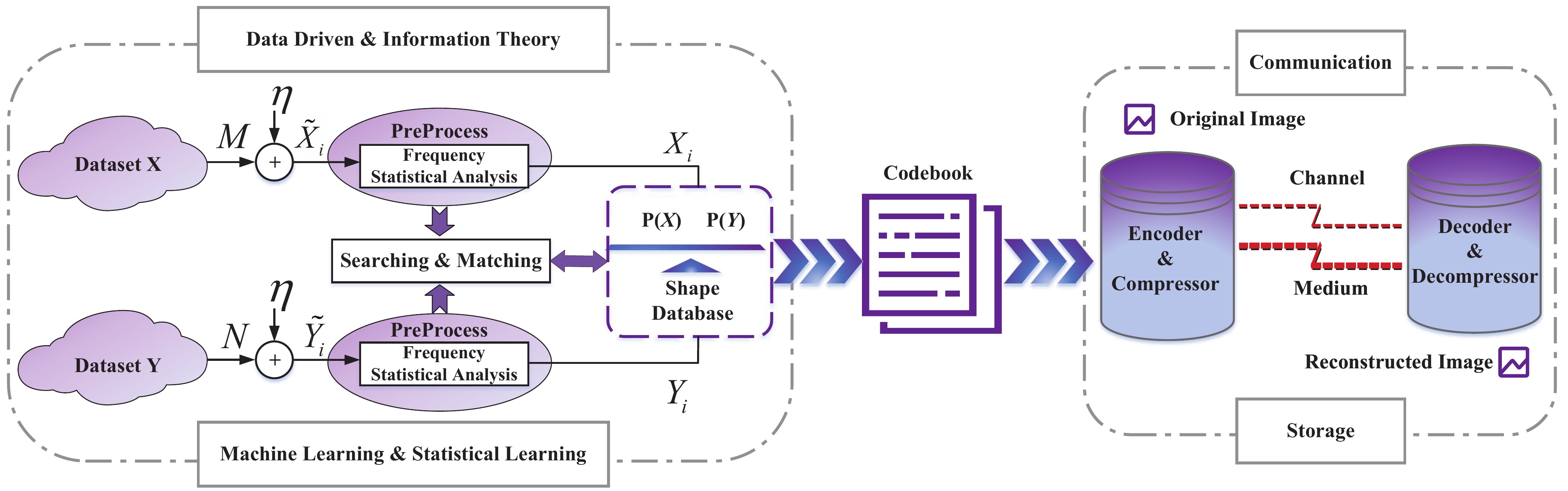}}
	\caption{The structure of shape coding. It consists of two parts: the generation and use of the codebook. The former makes use of the characteristics of the data source, while the latter improves the compression efficiency.}
	\label{Fig:framework}
\end{figure*}

\subsection{Relations to Previous Work}
The objective of this work is to present the performance limits from the viewpoint of information theory, which is related to our previous works in \cite{xin2020soft,xin2021soft,xin2021lossless}. An image encoding method through shapes and data-driven methods can improve image lossless compression. In some known databases, \emph{soft compression} outperforms the most popular methods such as PNG, JPEG2000, and JPEG-LS. However, there was no theoretical support for how shape-based \emph{soft compression} methods can reach the performance limit. That is, the gap between soft compression and its compression limit, namely the entropy rate is not theoretically known. On the other hand, the entropy rate associated with the asymptotic equipartition property analysis of images can help us design efficient encoding and decoding algorithms from the viewpoint of Shannon's information theory.

The earliest multi-pixel joint coding method can be traced back to Symbol-Based coding \cite{ascher1974a}, which transmits or stores only the first instance of each pattern class and thereafter substitutes this exemplar for every subsequent occurrence of the symbols. It achieved a degree of bandwidth reduction on a scan-digitized printed text. Fractal theory \cite{jacquin1992image,8683203} is also related to block-based coding. Fractal block coding approximates an original image by relying on the assumption that image redundancy can be efficiently exploited through self-transformability on a blockwise basis. On the other hand, \emph{soft compression} generates the shape database in a data-driven manner, so as to make the codebook used in the encoder and decoder. Image processing-based data-driven methods such as \cite{begaint2017region,zhang2020image,chen2019toward} can explore the essential features of images and even eliminate semantic redundancy. The method of using side information to assist data compression has also been used and analyzed by Kieffer \cite{yang2001universal} and Kontoyiannis \cite{gavalakis2021fundamental}. Verd\'{u} \cite{8283787} provided upper and lower bounds on the optimal guessing moments of a random variable by taking values on a finite set when the side information may be available. Rycht\'{a}rikov\'{a} \emph{et al.} \cite{rychtarikova2016point} generalized the point information gain and derived point information gain entropy, which may help analyze the entropy rate of an image.

Another connection to this paper is the Lempel-Ziv coding schemes \cite{ziv1978compression}. It proposed the concept of compressibility. For every individual infinite sequence $x$, a quantity $\rho(x)$ is defined. It is shown to be the asymptotically attainable lower bound on the compression ratio that can be achieved for $x$ by any finite-state encoder. Wyner \cite{wyner1989some} derived theorems concerning the entropy of a stationary ergodic information source and used the results to obtain insight into the workings of the Lempel-Ziv data compression algorithm. 

The main contribution of this paper is that we present a sufficient condition that allows us to show that the performance limit of shape-based image coding can be asymptotically achievable in terms of entropy rate.

\subsection{Paper Outline}

The rest of this paper is organized as follows. Section \ref{sec:asymptotic} contains our main results, giving asymptotic properties on shape-based image coding in terms of entropy rate. Moreover, it indicates the relationship between the number of shapes and coding performance. In Section \ref{sec:numerical}, we present sample numerical results with concrete examples. In Section \ref{sec:conclusion}, we give some complementary remarks and conclude this paper.

\section{The Asymptotic Properties of Image Sources Composed of Shapes}
\label{sec:asymptotic}

The encoding method with shapes can take advantage of the characteristics of the data and eliminate redundancy in the spatial and coding domains simultaneously. This section theoretically analyzes the performance of image coding with shapes. It will show that when the numbers of shapes and pixels have a reciprocal logarithm relationship, the average code length will asymptotically approach the entropy rate. To the best of our knowledge, this is the first result of image compression in information theory. The framework of this proof is similar to \cite{cover1999elements} \cite{wyner1989some}, but there are some important differences. 

The average number of bits needed to represent the image $Z$ with shapes is $B_Z$. Specifically,
\begin{align}
	\label{Eq:Bz-1}
	B_{Z} & ={1 \over t} \sum_{i=1}^{c(t)} [l(s_i)+l_p(x_i,y_i)] \\
	&\overset{(a)} \leq {{c(t)\log c(t) + \sum_{i=1}^{c(t)}l(s_i)}   \over t} \\ 
	& \overset{(b)} \leq { {c(t)\log c(t) + c(t)\log |\mathcal{D}|} \over t} \label{Eq:Bz-2}
\end{align}
where (a) and (b) follow from the fact that the uniform distribution has maximum entropy. That is, $\sum_{i=1}^{c(t)}l(s_i) \leq c(t) \log \mathcal{D}$ and $\sum_{i=1}^{c(t)}l_p(x_i,y_i) \leq c(t) \log c(t)$. $B_Z$ is the average cost of encoding $Z$, which reflects the coding requirements of bits. In the sequel, we use Eq. (\ref{Eq:Bz-2}) instead of (\ref{Eq:Bz-1}) to scale $B_Z$.

Let $\{Z_i\}_{i=-\infty}^{\infty}$ be a strictly stationary ergodic process with finite states and $ z_i^j \triangleq (z_i,z_{i+1},\ldots,z_j)$. Due to the invariance of time, $P(Z_t|Z_{t-k}^{t-1})$ is an ergodic process, where the $k$th-order Markov approximation is used to make an approximation. We will then have
\begin{equation}
	Q_k(z_{-(k-1)},\ldots,z_0,\ldots,z_t) \triangleq P(z_{-(k-1)}^0)\prod_{j=1}^t P(z_j|z_{j-k}^{j-1}) \label{eq:markov},
\end{equation}
where $z_{-(k-1)}^0$ is the initial state. In this way, one can use the $k$-th order Markov entropy rate to estimate the entropy rate of $\{Z_i\}$. That is,
\begin{align}
	-{1 \over t} \log Q_k(Z_1,Z_2,\ldots,Z_t|Z_{-(k-1)}^0) 
	&= -{1 \over t} \log \prod_{j=t}^t P(Z_j|Z_{j-k}^{j-1}) \\
	&= -{1 \over t} \sum_{j=1}^t \log P(Z_j|Z_{j-k}^{j-1}) \\
	&\to -E \log P(Z_j|Z_{j-k}^{j-1}) \\
	&=H(Z_j|Z_{j-k}^{j-1}).
\end{align}
When $k\to \infty$, the entropy rate of the $k$th-order Markov approximation converges to the entropy rate of the original random process.

Suppose that $z_1^t$ is decomposed into $c(t)$ shapes $s_1,s_2,\ldots,s_{c(t)}$. We define $w_i$ as the $k$ bits before $s_i$, where $w_1=z_{-(k-1)}^0$. Let $c_{lw}$ denote the number of shapes whose size is $l$ and its previous state $w_i=w$, $w\in \mathcal{Z}^k$.

\begin{lemma}
	\label{lemma:11}
	For $\{Z_i\}$, the joint transition probability and shape size satisfy the following inequality
	\label{lemma:1}
	\begin{equation}
		\log Q_k (z_1,z_2,\ldots,z_t|w_1) \leq \sum_{l,w} c_{lw} \log {\alpha \over c_{lw}},
	\end{equation}
\end{lemma}
where $\alpha$ is a constant.

\begin{proof}
	Suppose that for fixed $l$ and $w$, the sum of the transition probabilities is less than a constant $\alpha$, i.e., 
	\begin{equation}
		\label{Eq:Suppose A}
		\sum_{i:|s_i|=l,w_i=w} {1 \over c_{lw}} P(s_i|w_i) \leq \alpha.
	\end{equation}
	Then,
	\begin{align}
		\log Q_k (z_1,z_2,\ldots,z_t|w_1) &= \log Q_k(s_1,s_2,\ldots,s_c|w_1) \\
		&\overset{(a)} = \sum_{i=1}^c \log P(s_i|w_i) \\
		&= \sum_{l,w} \sum_{i:|s_i|=l,w_i=w} \log P(s_i|w_i) \\
		&= \sum_{l,w} c_{lw} \sum_{i:|s_i|=l,w_i=w} {1 \over c_{lw}} \log P(s_i|w_i) \\
		& \overset{(b)} \leq \sum_{l,w} c_{lw} \log \sum_{i:|s_i|=l,w_i=w} {1 \over c_{lw}} P(s_i|w_i) \\
		& \leq \sum_{l,w} c_{lw} \log {\alpha \over c_{lw}}
	\end{align}
	where (a) follows from Eq. (\ref{eq:markov}) and (b) follows from Jensen's inequality, thanks to the convexity of $\log x$ for $x>0$.
\end{proof}

Lemma \ref{lemma:11} links the conditional probability $Q_k(z_1,z_2,...,z_t|w_1)$ to $c_{lw}$, connecting the concepts before and after decomposing $\{ Z_i\}$. We will continue to explore the quantitative relationship between shapes and pixels.
\begin{lemma}
	For $\{Z_i\}$, the number and size of its shapes meet the following relationship
	\label{lemma:2}
	\begin{equation}
		\sum_{l,w} c_{lw} \log {c \over c_{lw}} \leq kc+(t+c)\log (1 + {c \over t}) + c \log {t \over c}
	\end{equation}
\end{lemma}

\begin{proof}
	For simplicity, we use $c$ to represent $c(t)$. Let $p_{lw} = \displaystyle{\frac{c_{lw}}{c}} $, then $\sum\limits_{l,w}p_{lw}=1$. 
	We define two random variables $U$ and $V$ such that
	\begin{equation}
		\Pr (U=l, V=w) = p_{lw}.
	\end{equation}
	The mean of $U$ is the average length of each shape, i.e., $E(U)=\displaystyle{\frac{t}{c}}$. A random variable with a geometric distribution has maximum entropy when the mean of a discrete random variable is fixed. Thus, we have,
	\begin{align}
		H(U) &\overset{(a)}\leq {t \over c} \log {t \over c}-({{t \over c}-1}) \log ({t \over c}-1) \\
		&\overset{(b)}\leq ({{t \over c}+1}) \log ({t \over c}+1) - {t \over c} \log {t \over c} \\
		& = \log {t \over c} + ({t \over c} + 1)\log ({c \over t} + 1),
	\end{align}
	where (a) is the entropy of a random variable with a geometric distribution and (b) follows that the function $f(x)=x\log x -(x-1)\log (x-1)$ is monotonically increasing when $x \geq 1$. On the other hand,      
	$H(V)\leq \log |\mathcal{Z}|^k = k$. Thus, 
	\begin{align}
		\sum_{l,w} c_{lw} \log {c \over c_{lw}} &= c \sum_{l,w}p_{lw} \log {1 \over p_{lw}} \\
		&= cH(U,V)  \\
		&\leq cH(U) + cH(V) \\
		&\leq c[\log {t \over c} + ({t \over c} + 1)\log ({c \over t} + 1)] + kc \\
		& = kc+(t+c)\log (1 + {c \over t})+c\log{t \over c},
	\end{align}
	which completes the proof.
\end{proof}
Based on these two lemmas, we will further analyze the condition under which the entropy rate can be reached asymptotically.

\begin{Theorem}
	\label{theorem:1}
	When the numbers of shapes and pixels meet the reciprocal relation of the logarithm, then the average encoding length will asymptotically approximate the entropy rate. That is, \\
	if
	\begin{equation}
		\label{Eq:relationship}
		{c(t) \over t} = \mathcal{O} ({1 \over {\log t}})
	\end{equation}
	then	
	\begin{equation}
		\label{Eq:result}
		\lim\limits_{t \to \infty} {l(Z_1,Z_2,...,Z_t) \over t} = H(\mathcal{Z}).
	\end{equation}	
\end{Theorem}

\begin{proof}
	From Lemma \ref{lemma:1}, one can write
	\begin{align}
		\log Q_k (z_1,z_2,...,z_t|w_1) & \leq \sum_{l,w} c_{lw} \log {\alpha \over c_{lw}}  \\
		&= -\sum_{l,w}c_{lw}\log {{c \cdot c_{lw}} \over {c\cdot \alpha}} \\
		&= -c \log c-\sum_{l,w} c_{lw} \log {c_{lw} \over {c\alpha}}
	\end{align}
	
	For simplicity, we use $Q$ to represent $Q_k (z_1,z_2,...,z_t|w_1)$. Thus,
	\begin{align}
		{c \over t} \log c \leq -{1 \over t} \log Q - {1 \over t} \sum_{l,w} c_{lw} \log {c_{lw} \over c \alpha} \label{Eq:middle}
	\end{align}	
	
	From Lemma \ref{lemma:2}, it follows that
	\begin{align}
		-{1 \over t} \sum_{l,w} c_{lw} \log {c_{lw} \over c \alpha} &= {1 \over t} \sum_{l,w} c_{lw}\log {c \over c_{lw}} + {c \over t} \log \alpha\\
		&\leq {1 \over t} [kc+(t+c)\log (1 + {c \over t})+c\log{t \over c}] + {c \over t} \log \alpha \\
		&={c \over t}(k+\log \alpha) + {c \over t } \log {t \over c} + (1 + {c \over t})\log (1 + {c \over t})  \label{Eq:three component}.
	\end{align}
	
	When ${c \over t} = \mathcal{O} ({1 \over {\log t}})$ and $t \to \infty$, the three terms in the right-hand side of Eq. (\ref{Eq:three component}) will all tend to 0. Combining Eq. (\ref{Eq:middle}) and (\ref{Eq:three component}), we obtain
	\begin{equation}
		-{1 \over t} \sum_{l,w} c_{lw} \log {c_{lw} \over c \alpha} \to 0,~~\textit{when}~t\to \infty.
	\end{equation}
	
	Then,
	\begin{align}
		\lim\limits_{t\to \infty}  \sup {{c(t) \log c }\over t} 
		&\leq \lim\limits_{t \to \infty} {-{1 \over t}\log Q_k(Z_1,Z_2,...,Z_t|Z_{-(k-1)}^0}) \\
		&\to H(\mathcal{Z}).
	\end{align}
	
	The asymptotic property of the second term in the right-hand side of Eq. (\ref{Eq:Bz-2}),
	\begin{align}
		\lim\limits_{t \to \infty} {{c(t) \log |\mathcal{D}|} \over t} = 0.
	\end{align}
	
	Thus, 
	\begin{align}
		\lim\limits_{t \to \infty} {l(Z_1,Z_2,...,Z_t) \over t} &= \lim\limits_{t \to \infty} ({c(t)\log c \over t} + {c(t) \log |\mathcal{D}| \over t}) \\
		&= \lim\limits_{t \to \infty} {c(t)\log c \over t} \\
		&= H(\mathcal{Z}).
	\end{align}
	
	This shows that when $c(t)$ and $t$ meet the condition in Eq. (\ref{Eq:relationship}), the average coding length of $\{Z_i\}$ will asymptotically approximate the entropy rate $H(\mathcal{Z})$.
\end{proof}
Theorem \ref{theorem:1} sets up a bridge between the shapes and the entropy rate for image sources with ergodic properties. It theoretically indicates what order of magnitude we should have the shapes and pixels. When one encodes images with shapes, the average cost will asymptotically tend to the entropy rate if the numbers of shapes and pixels satisfy the reciprocal relation of the logarithm. Moreover, it gives new insights into designing image compression algorithms in theory. 

\section{Numerical Analysis}
\label{sec:numerical}

\begin{table}
	\renewcommand\arraystretch{1.2}
	\caption{The numerical analysis of shapes and pixels on MNIST dataset ($R_{avg}$ is the average compression ratio)}
	\setlength{\tabcolsep}{7pt}
	\small
	\begin{center}
		\begin{tabular}{ccccccccccc}
			\toprule
			Class& 0    & 1    & 2    & 3    & 4   & 5 & 6    & 7    & 8    & 9\\
			\midrule
			$R_{avg}$ & 2.84 & 6.02 & 3.17 & 3.20 & 3.77 & 3.40 & 3.20 & 4.05 & 2.81 & 3.52   \\
			${1 \over t}{c(t)\log t}$        & \textbf{0.200} & \textbf{0.080} & \textbf{0.178} & \textbf{0.175} & \textbf{0.149} & \textbf{0.163} & \textbf{0.175} & \textbf{0.136} & \textbf{0.202} & \textbf{0.157}    \\
			\bottomrule
		\end{tabular}
	\end{center}
	\label{Tab:analysis-1}
\end{table}

Section \ref{sec:asymptotic} points out the asymptotic property of encoding methods based on shapes. When ${c(t) \over t} \to \mathcal{O}({1 \over \log t})$, the average encoding length will asymptotically approximate the entropy rate. It indicates the relationship between the shape-pixel number ratio and coding performance. In this section, we present some numerical results to illustrate that for each ergodic process of an image source, if ${c(t) \over t} \to  \mathcal{O}({1 \over \log t})$ as $t \to \infty$, one can get the result of Eq. (\ref{Eq:result}).

Table \ref{Tab:analysis-1} reveals the numerical results of MNIST datasets. It includes encoding results $R_{avg}$ and ${1 \over t}c(t)\log t$ in ten categories with the soft compression algorithm \cite{xin2020soft}. What can be clearly seen in this table is that ${1 \over t}c(t)\log t <1$ for all classes. It is on the order of $\mathcal{O}(1)$, which is consistent with the assumption in Theorem \ref{theorem:1}.

Now we focus on simulated images as an alternative analysis. We use the birth and death processes of two states to simulate a stationary ergodic process. For each case, 5000 $\{Z_i\}$ with $M=100,N=100$ are generated, respectively. We encode $\{ Z_i\}$ with fixed size shapes and observe the effect of ${c(t) \over t}$ on coding performance. 

Fig. \ref{Fig:curve} illustrates the shape coding working mechanism of the image source. It indicates the performance of the encoding method with shapes, in bits per pixel (bpp). Cases 1-5 represent different parameters of the infinitesimal generator matrix of the birth-death process, illustrating the relationship between coding performance and {${c(t)} \over t$}. In different cases, the changing trend of these curves is the same. The bpp decreases with the increase of shape size (i.e., the shape-pixel number ratio decreases), which reflects the gain brought by shape. Moreover, as the shape-pixel number ratio continues to decrease, bpp enters the smoothing region. It also shows that the reduction of the number ratio will not always improve the encoding performance. This is due to the fineness of the model itself, which does not take advantage of the additional statistical information of larger shapes. Note that, the numerical difference between the curves is essentially the difference in the entropy rate.

\begin{figure}
	\centerline{\includegraphics[width=3.5in]{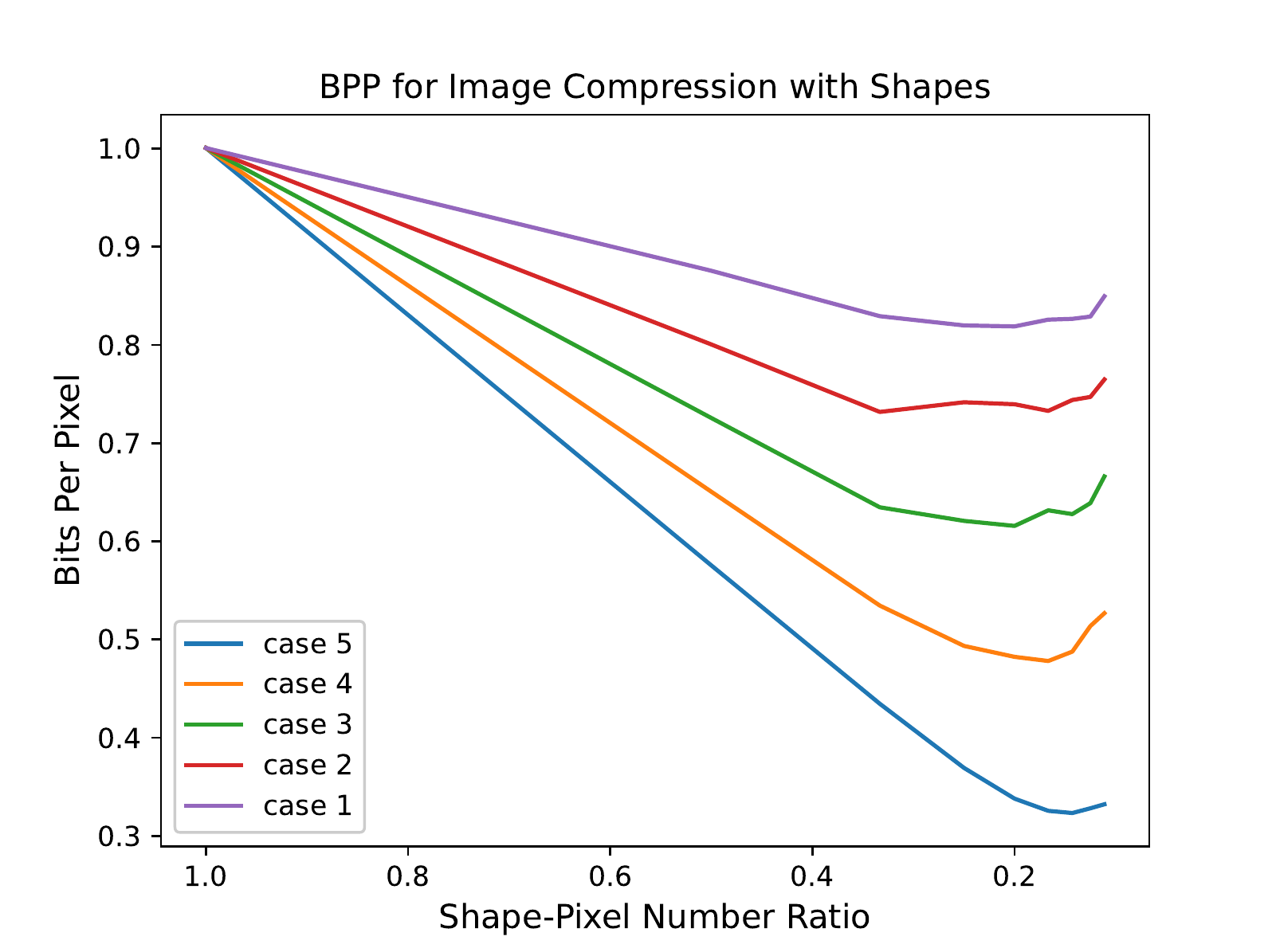}}
	\caption{The performance of encoding method with shapes, in bits per pixel (bpp).}
	\label{Fig:curve}
\end{figure}

\section{Concluding Remarks}
\label{sec:conclusion}
In this paper, we investigated the performance limit of shape-based image compression. Our works answered the open problem of the relationship between image decomposition and lossless compression, which reflects the performance variation in general. Specifically, when the numbers of shapes and pixels have a reciprocal relation to the logarithm, the average code length will asymptotically approach the entropy rate.

For image coding algorithms, one should give full attention to the superiority of shapes in image processing. Likewise, it is necessary to take advantage of the characteristics of the image dataset. Through shapes and data-driven methods, one can use the high-dimensional information of images to help code. Moreover, the asymptotic analysis of the entropy rate can also be extended to gray images and multi-component images with some adjustments.

Finally, it should be noted that this paper focuses on the source part, without considering the natural robustness of images in the communication process. In future work, we will explore the theory of joint source-channel image coding in the finite block length regime. It is noted that image lossless compression, especially, soft compression, may become an important block for semantic information communications, and even play some roles in the new developments of metaverse-type services in the future.

\bibliographystyle{IEEEtran}
\bibliography{reference}

\begin{thebibliography}{10}
\providecommand{\url}[1]{#1}
\csname url@samestyle\endcsname
\providecommand{\newblock}{\relax}
\providecommand{\bibinfo}[2]{#2}
\providecommand{\BIBentrySTDinterwordspacing}{\spaceskip=0pt\relax}
\providecommand{\BIBentryALTinterwordstretchfactor}{4}
\providecommand{\BIBentryALTinterwordspacing}{\spaceskip=\fontdimen2\font plus
\BIBentryALTinterwordstretchfactor\fontdimen3\font minus
  \fontdimen4\font\relax}
\providecommand{\BIBforeignlanguage}[2]{{%
\expandafter\ifx\csname l@#1\endcsname\relax
\typeout{** WARNING: IEEEtran.bst: No hyphenation pattern has been}%
\typeout{** loaded for the language `#1'. Using the pattern for}%
\typeout{** the default language instead.}%
\else
\language=\csname l@#1\endcsname
\fi
#2}}
\providecommand{\BIBdecl}{\relax}
\BIBdecl

\bibitem{9376651}
Y.~Hu, W.~Yang, Z.~Ma, and J.~Liu, ``Learning end-to-end lossy image
  compression: A benchmark,'' \emph{IEEE Transactions on Pattern Analysis and
  Machine Intelligence}, vol.~44, no.~8, pp. 4194--4211, 2022.

\bibitem{choi2022semantic}
J.~Choi and J.~Park, ``Semantic communication as a signaling game with
  correlated knowledge bases,'' \emph{arXiv preprint arXiv:2209.00280}, 2022.

\bibitem{xin2022exk}
G.~Xin and P.~Fan, ``Exk-sc: A semantic communication model based on
  information framework expansion and knowledge collision,'' \emph{Entropy},
  vol.~24, no.~12, p. 1842, 2022.

\bibitem{cover1999elements}
T.~M. Cover, \emph{Elements of information theory}.\hskip 1em plus 0.5em minus
  0.4em\relax John Wiley \& Sons, 1999.

\bibitem{xin2020soft}
G.~Xin, Z.~Li, Z.~Zhu, S.~Wan, P.~Fan, and K.~B. Letaief, ``Soft compression:
  An approach to shape coding for images,'' \emph{IEEE Communications Letters},
  vol.~25, no.~3, pp. 798--801, 2020.

\bibitem{xin2021soft}
G.~Xin and P.~Fan, ``Soft compression for lossless image coding based on shape
  recognition,'' \emph{Entropy}, vol.~23, no.~12, p. 1680, 2021.

\bibitem{xin2021lossless}
------, ``A lossless compression method for multi-component medical images
  based on big data mining,'' \emph{Scientific Reports}, vol.~11, no.~1, pp.
  1--11, 2021.

\bibitem{ascher1974a}
R.~{Ascher} and G.~{Nagy}, ``A means for achieving a high degree of compaction
  on scan-digitized printed text,'' \emph{IEEE Transactions on Computers},
  vol.~23, no.~11, pp. 1174--1179, 1974.

\bibitem{jacquin1992image}
A.~{Jacquin}, ``Image coding based on a fractal theory of iterated contractive
  image transformations,'' \emph{IEEE Transactions on Image Processing},
  vol.~1, no.~1, pp. 18--30, 1992.

\bibitem{8683203}
F.~Guerrini, A.~Gnutti, and R.~Leonardi, ``Iterative mirror decomposition for
  signal representation,'' in \emph{ICASSP 2019 - 2019 IEEE International
  Conference on Acoustics, Speech and Signal Processing (ICASSP)}, 2019, pp.
  5541--5545.

\bibitem{begaint2017region}
J.~B{\'e}gaint, D.~Thoreau, P.~Guillotel, and C.~Guillemot, ``Region-based
  prediction for image compression in the cloud,'' \emph{IEEE Transactions on
  Image Processing}, vol.~27, no.~4, pp. 1835--1846, 2017.

\bibitem{zhang2020image}
X.~Zhang, C.~Yang, X.~Li, S.~Liu, H.~Yang, I.~Katsavounidis, S.-M. Lei, and
  C.-C.~J. Kuo, ``Image coding with data-driven transforms: Methodology,
  performance and potential,'' \emph{IEEE Transactions on Image Processing},
  vol.~29, pp. 9292--9304, 2020.

\bibitem{chen2019toward}
Z.~Chen, L.-Y. Duan, S.~Wang, Y.~Lou, T.~Huang, D.~O. Wu, and W.~Gao, ``Toward
  knowledge as a service over networks: A deep learning model communication
  paradigm,'' \emph{IEEE Journal on Selected Areas in Communications}, vol.~37,
  no.~6, pp. 1349--1363, 2019.

\bibitem{yang2001universal}
E.-h. Yang, A.~Kaltchenko, and J.~C. Kieffer, ``Universal lossless data
  compression with side information by using a conditional mpm grammar
  transform,'' \emph{IEEE Transactions on Information Theory}, vol.~47, no.~6,
  pp. 2130--2150, 2001.

\bibitem{gavalakis2021fundamental}
L.~Gavalakis and I.~Kontoyiannis, ``Fundamental limits of lossless data
  compression with side information,'' \emph{IEEE Transactions on Information
  Theory}, vol.~67, no.~5, pp. 2680--2692, 2021.

\bibitem{8283787}
I.~Sason and S.~Verdú, ``Improved bounds on lossless source coding and
  guessing moments via rényi measures,'' \emph{IEEE Transactions on
  Information Theory}, vol.~64, no.~6, pp. 4323--4346, 2018.

\bibitem{rychtarikova2016point}
R.~Rycht{\'a}rikov{\'a}, J.~Korbel, P.~Mach{\'a}{\v{c}}ek, P.~C{\'\i}sa{\v{r}},
  J.~Urban, and D.~{\v{S}}tys, ``Point information gain and multidimensional
  data analysis,'' \emph{Entropy}, vol.~18, no.~10, p. 372, 2016.

\bibitem{ziv1978compression}
J.~Ziv and A.~Lempel, ``Compression of individual sequences via variable-rate
  coding,'' \emph{IEEE transactions on Information Theory}, vol.~24, no.~5, pp.
  530--536, 1978.

\bibitem{wyner1989some}
A.~D. Wyner and J.~Ziv, ``Some asymptotic properties of the entropy of a
  stationary ergodic data source with applications to data compression,''
  \emph{IEEE Transactions on Information Theory}, vol.~35, no.~6, pp.
  1250--1258, 1989.

\end{thebibliography}

\vspace{-33pt}

\newpage

\vfill

\end{document}